\documentclass{elsart}

\journal{ArXiv quant-ph}

\usepackage{graphicx}

\usepackage{amssymb}
\usepackage{amsmath}

\renewcommand{\half}{\frac{1}{2}}

\def\k2sum { \sum_{ k \in (-\pi/2,\pi/2) } }
\def\kpi2sum { \sum_{ k \in (-\pi/2,\pi/2) } }

\def\kpi2sigmasum { \sum_{\stackrel{k\in(-\pi/2,\pi/2) }{ \sigma = \pm } } }

\newtheorem{lemma}{Lemma}[section]
\newtheorem{theorem}[lemma]{Theorem}

\theoremstyle{definition}

\theoremstyle{remark}
\numberwithin{equation}{section}

\newenvironment{proof}{{\it Proof.}}{$\Box$}

\begin{document}

\begin{frontmatter}

\title{Absorption Probabilities for the Two-Barrier Quantum Walk}

\author[1]{Eric~Bach\corauthref{cor}}\ead{bach@cs.wisc.edu},
\author[2]{Lev Borisov}\ead{borisov@math.wisc.edu},

\address[1]{
  Computer Sciences Department, University of Wisconsin\\
  1210 W.~Dayton St., Madison, WI~~53706
}

\address[2]{
  Department of Mathematics, University of Wisconsin\\
  480 Lincoln Dr., Madison, WI~~53706
}

\corauth[cor]{Corresponding author. Phone 608-262-1204, Fax 608-262-9777.}

\date{\today}
\maketitle

\begin{abstract}
Let $p_j^{(n)}$ be the probability that a Hadamard quantum walk,
started at site $j$ on the integer lattice $\{0,\ldots,n\}$, is
absorbed at 0.  We give an explicit formula for $p_j^{(n)}$.
Our formula proves a conjecture of John Watrous, concerning an
empirically observed linear fractional recurrence relation for
the numbers $p_1^{(n)}$.
\end{abstract}

\end{frontmatter}

\section{Introduction}
\label{sec:introduction}

Consider a Hadamard quantum walk on the sites
$0,1,\ldots,n$, as defined in \cite{ABNVW}.   The boundary sites,
0 and $n$, are absorbing, so any walk is
certain to be absorbed.  Let $p_j^{(n)}$ denote the probability 
that the walk ends at location 0.  

The main result of this paper is an explicit formula for this
probability.  Let $A = 2 + \sqrt 2$, and $B = 2 - \sqrt 2$.
We will prove that when $n \ge 1$ and $1 \le j \le n$, we have
\begin{equation}
\label{eqn:explicit}
p_j^{(n)} = \frac{\sqrt 2}{4}
            \frac{ ( A^{n-j} - B^{n-j} ) ( A^{j-1} + B^{j-1} ) }
                 { A^{n-1} + B^{n-1} },
\end{equation}
In particular, when $j=1$,
\begin{equation}
\label{eqn:explicitj1}
p_1^{(n)} = \frac {\sqrt 2}{2}
            \frac{ A^{n-1} - B^{n-1} }
                 { A^{n-1} + B^{n-1} },
\end{equation}
and this satisfies the recurrence relation
\begin{equation}
\label{eqn:watrous}
p_1^{(n)} = \frac{1 + 2 p_1^{(n-1)}}{2 + 2 p_1^{(n-1)}}.
\end{equation}
This recurrence relation, conjectured by Watrous from numerical data,
appears in \cite{ABNVW}.  It is a consequence of 
(\ref{eqn:explicitj1}) that
$$
\lim_{n \to \infty} p_1^{(n)} =  \frac 1 {\sqrt 2},
$$
a result that was proved with some effort in \cite{ABNVW}.

The absorption probability values are interrelated in many interesting ways.
In the two-dimensional table of $p_j^{(n)}$, there is
a linear fractional recurrence relation, similar to (\ref{eqn:watrous}),
for each column and each diagonal.  This ``numerology'' was first
observed empirically, and then led us to conjecture (\ref{eqn:explicit}).
There is also a linear
recurrence relation common to all rows: if $1 \le j \le n-3$, we have
$$
p_j^{(n)} - 7 p_{j+1}^{(n)} + 7 p_{j+2}^{(n)} - p_{j+3}^{(n)} = 0.
$$
Two other relations, which we discuss later, have combinatorial 
interpretations.

It is interesting to consider the implications of (\ref{eqn:explicit}) 
for starting sites in the ``interior'' of the lattice.  
For $j= \alpha n$, $\alpha$
fixed and $n \rightarrow \infty$, $p_j^{(n)}$ will be close to the
limit $\frac {\sqrt 2} 4 = 0.35355...$\ .  Thus, from the interior region,
the probabilities for absorption at the left and right are almost
constant, approximately 35\% and 65\%.  On the other hand, for the classical
random walk, with equal probability of moving left and right, the
probability that the walk, starting from $j$, is absorbed at the left,
decreases linearly with $j$,
from 1 at the left barrier to 0 at the right barrier.
(One proof of this appears in \cite{DS}.)  This gives another example
of the idea that quantum walks ``spread out'' more evenly than 
classical walks do.

\section{Some Generating Functions.}

Our work will be based on the path count generating function
approach which was employed in \cite{ABNVW} and \cite{BCGJW}.

Recall that for our walk, each site $j$ has two states, 
corresponding to the walker facing left and facing right.
The one-step evolution matrix is a Hadamard transformation,
i.e. starting from state $|n,R\rangle$ the particle can go next to
\begin{eqnarray*}
        |n+1, R\rangle & \hbox{ with amplitude } 1/\sqrt 2  \\
        |n-1, L\rangle & \hbox{ w. a. } 1/\sqrt 2  \\
\end{eqnarray*}
and from $|n,L\rangle$ it can go to
\begin{eqnarray*}
        |n+1, R\rangle & \hbox{ w.a. } 1/\sqrt 2  \\
        |n-1, L\rangle & \hbox{ w.a. } - 1/\sqrt 2 
\end{eqnarray*}
The initial state of the particle is $|j,R \rangle$, with $0 < j < n$.

We define the {\sl path count generating function} to be
\begin{equation}
\label{eqn:PCGF}  
f_j^{(n)}(z) = \sum_{m \ge 1} \left( \sum_{P} \sigma(P) \right) z^m
\end{equation}
where the $m$-th sum is over paths $P$ of length $m$ that
are absorbed at the left (0) state,
and the sign, $\sigma(P)$, is $(-1)^{\hbox{ \# of $LL$ blocks }}$.  
In computing
the sign, overlaps count, for example, the sign of $LLL$ is +1.

The probability of absorption at 0 is given by the formula
\begin{equation}
\label{eqn:pjint}
p_j^{(n)} 
= {1 \over 2 \pi i} \int_{|z| = 1/\sqrt 2} |f_j^{(n)}(z)|^2  z^{-1} dz,
\end{equation}

In \cite{ABNVW} it was shown that
\begin{equation}
f_1^{(2)} = z
\end{equation}
and for $n > 1$,
\begin{equation}
\label{eqn:f1n}
f_1^{(n)} = z {1 - 2z f_1^{(n-1)} \over 1 - z f_1^{(n-1)} }.
\end{equation}
We now complement this with another recurrence relation that allows 
(\ref{eqn:PCGF}) to be computed for $j>1$.  
Observe that any path from $j$ that is absorbed
at 0 must go through 1.  Hence any such absorbed path breaks up into:
a) a path from $j$ to 1 reaching 1 only once; and b) a path
from 1 to 0.  Part a) is of the same shape as a path from $j-1$ to 0
on a lattice with absorption at $n-1$, and part b) is just a path from
1 to 0.  The path of part b), if it immediately moves left, must
be preceded by an $L$ move, so we must correct the sign for this case
(and this case only).  This gives
\begin{equation}
\label{eqn:fjn}
f_j^{(n)} = f_{j-1}^{(n-1)} \left( f_1^{(n)} - 2z \right),
\qquad \qquad 2 \le j < n.
\end{equation}

Here are few of these functions.
$$
f_1^{(3)} = {z (2z^2 - 1) \over (z-1)(z+1)},
\qquad
f_2^{(3)} = {z^2 \over (z-1)(z+1)};
$$
$$
f_1^{(4)} = {z (4z^4 - 3z^2 + 1) \over 2z^4 - 2z^2 + 1},
\qquad
f_2^{(4)} = {z^2 (2z^2 - 1) \over 2z^4 - 2z^2 + 1},
\qquad
f_3^{(4)} = {z^3 \over 2z^4 - 2z^2 + 1};
$$
$$
f_1^{(5)} = {z (2z^2 - 1) (4z^4 - 2z^2 + 1)
            \over (2z^3 + z^2 - z - 1)(2z^3 - z^2 - z + 1)},
f_2^{(5)} = {z^2 (4z^4 - 3z^2 + 1)
            \over (2z^3 + z^2 - z - 1)(2z^3 - z^2 - z + 1)},
$$
$$
f_3^{(5)} = {z^3 (2z^2 - 1)
            \over (2z^3 + z^2 - z - 1)(2z^3 - z^2 - z + 1)},
f_4^{(5)} = {z^4
            \over (2z^3 + z^2 - z - 1)(2z^3 - z^2 - z + 1)}.
\qquad
$$

We note that for each $n$, the $f_j^{(n)}$ have the same denominator.  This
can be proved as follows.  We first observe that the power series for
$f_j$ begins with $\pm z^j$, so that $f_j = z^j u$, with $u(0) \ne 0, \infty$.
In particular, we have $f_1 = z a(z)/b(z)$, where $a$ and $b$ are polynomials.
Next, we group paths according to the location
of the first $R$ move and find the relation
$$
f_j = z f_{j+1} + \sum_{k=2}^j (-)^k z^k f_{j+2-k}
                           + (-)^{j-1} z^j,
$$
valid for $1 \le j < n-1$.  From this it follows that $b f_{j+1}$ is
a polynomial, using induction on $j$.

It is interesting that the path sign is essentially the same
as the Rudin-Shapiro coefficient.  This coefficient $a_n$ is determined
by the parity of the number of ``11'' blocks in the binary 
notation of the positive
integer $n$.  The Rudin-Shapiro coefficient has many applications,
including the solution of extremal problems in classical Fourier 
analysis.  For a survey of this topic, see \cite{MF}.

\section{Some Combinatorial Results.}
\label{sect:combin}

It is interesting to see how much can be determined by purely
combinatorial arguments, without relying on integration.  
We begin with two interesting relations.

\begin{theorem}
\label{thm:firstlast}
We have
$$
p_1^{(n)} + p_{n-1}^{(n)} = 1.
$$
That is, in any row, the outer entries sum to 1.
\end{theorem}

\begin{proof}
Let $P$ be a path that starts
from 1 and is absorbed at $n$.  Its complement $\bar P$, obtained
by interchanging $L$ and $R$, is a path that starts at $n-1$
and is absorbed at 0.  We note that $P$ must begin and end
with $R$ moves. Let $P$ contain $\ell$ $L$ moves, $r$ $R$ moves,
and have $k$ occurrences of $RL$.  Then the sign of $P$
is $(-1)^{\ell - k}$.  However,
$$
k = \hbox{ \# of $RL$ in $\bar P$ }
  = \hbox{ \# of $LR$ in $\bar P$ },
$$
since $\bar P$ begins and ends with $L$ moves.  Therefore we have
$$
\sigma(P) = (-1)^{\ell - k}, \qquad \sigma(\bar P) = (-1)^{r - k},
$$
which implies
$$
\sigma(P) \sigma(\bar P) = (-1)^{\ell + r}
                         = (-1)^{r - \ell}
                         = (-1)^{n-1}.
$$
This tells us that if we complement all paths, the probability
is unchanged, since it is a sum of squares of quantities (signed
path counts) that individually change only by a sign.  So
$$
1 - p_1^{(n)} = \Pr[\hbox{ a walk from 1 is absorbed at $n$ }]
$$
$$
= \Pr[\hbox{ a walk from $n-1$ is absorbed at 0 }] = p_{n-1}^{(n)}.
$$
\end{proof}

If we let $q_j^{(n)} = 1 - p_j^{(n)}$ be the probability that the
walk reaches site $n$, then Theorem \ref{thm:firstlast}
looks like an ``obvious'' symmetry relation that should also hold
for $j>1$.  However, this is not so.  (See Table 1 at the end
of this paper.)

\begin{theorem}
\label{thm:firstsecond}
We have
$$
2 p_1^{(n)} = p_2^{(n)} + 1.
$$
In words, doubling the first number in any row is the same as
increasing the second by 1.
\end{theorem}

\begin{proof}
Observe that
$$
f_1^{(n)}(z) = z + z f_2^{(n)}(z).
$$
Now take the Hadamard square of both sides, and evaluate at $1/2$.
The two terms in the right side do not interfere because $f_2^{(n)}(z)$
is a multiple of $z^2$.
\end{proof}

For $n = 2,3$, the absorption probabilities can be obtained by combinatorial
reasoning, without any integration.  First, for $n=2$, we have
$$
p_1^{(2)} = 1/2,
$$
since there are only two possible paths, each absorbed after one
step.  To compute $p_1^{(3)}$, observe that for $t = 1, 3, 5, \ldots$,
there is precisely one path that reaches site 0 after $t$ steps.  Therefore,
the signs are irrelevant, and we have
$$
p_1^{(3)} = \frac 1 2 + \frac 1 8 + \frac 1 {32} + \cdots = \frac 2 3 .
$$
Using Theorem \ref{thm:firstlast}, we find 
$$
p_2^{(3)} = \frac 1 3 .
$$

\section{Proof of the Explicit Formula.}

In this section we prove that (\ref{eqn:explicit}) holds.
What we would like to do is integrate (\ref{eqn:pjint}) by residues,
and expose the dependence of $p_j^{(n)}$ on $n$.  With the original
integral, this is probably impossible, since we do not know where
the poles inside the circle of integration actually are.  
However, we can express $p_j^{(n)}$ using the integral of
a new rational function, with the same mysterious poles inside, 
but with only one pole outside.
Since the sum of the residues of any rational function vanishes,
we can just as well evaluate the residue outside the circle,
and this leads to a formula for $p_j^{(n)}$.

Let us begin with another formula for the path count generating
function.  Let $\alpha, \beta$ be the two roots of
$$
T^2 - (1 - 2z^2)T - z^2 = 0.
$$
These will always be used symmetrically so we do not care which is which.
Explicitly,
$$
\alpha,\beta = \frac{1 - 2z^2 \pm \sqrt{1 + 4z^4} } 2 .
$$
Also there is a recurrence relation
\begin{equation}
\label{eqn:alphabetarec}
\alpha^k = (1 - 2z^2) \alpha^{k-1} + z^2 \alpha^{k-2}
\end{equation}
and similarly for $\beta$.


\begin{lemma}
If $1 \le j < n$, then
$$
f_j^{(n)} = z^j \frac 
             { \alpha^{n-j} -  \beta^{n-j} }
             { \alpha^n -  \beta^n + 2z^2 (\alpha^{n-1} -  \beta^{n-1}) },
$$
\end{lemma}

\begin{proof}
This can be proved by induction.  First, let $j=1$ and increase $n$,
using (\ref{eqn:f1n}).   Then, for each $n$ in turn, let
$j = 2,\ldots,n-1$, and use (\ref{eqn:fjn}).
\end{proof}

It will be convenient to allow let $j=n$, and define $f_n^{(n)} = 0$,
so that $p_n^{(n)} = 0$.  This is consistent with the above formula, 
as well as (\ref{eqn:fjn}).

As a function of $z$, $f_j{(n)}$ is odd or even according as $j$ is.
We now write $f_j^{(n)} = z^j g_j^{(n)}$, and $t = z^2$.  Then,
$$
g_j^{(n)}(t) = \frac {r_{n-j}} { r_n - 2t r_{n-1}},
$$
where
$$
r_k = \frac {\alpha^k - \beta^k} {\alpha - \beta}.
$$
Since $r_{k+2}-(\alpha+\beta)r_{k+1}+\alpha\beta r_k = 0$,
we have $r_0 = 0$, $r_1 = 1$, and
\begin{equation}
\label{rec}
r_{k+2}-(1-2t)r_{k+1}- t r_k = 0
\end{equation}
Consequently, the $r_k$ are polynomials in $t$.

On the circle $|z| = 1/ \sqrt 2$, we have $\bar z = 1/(2z)$, and
$\bar t = 1/(4t)$.  Making this substitution into $\bar g_j$,
clearing denominators, and observing that $\alpha \beta = -t$, we get
$$
| g_j^{(n)} | ^2
= (-)^j 2^j t^j 
  \frac{ r_{n-j}^2 }{ ( r_n + 2t r_{n-1} ) (r_n - r_{n-1}) }
$$
Since $f_j^{(n)} = z^j g_j^{(n)}$, we get from this
\begin{equation}
\label{eqn:integral}
p_j^{(n)} 
= \frac{(-1)^j}{2 \pi i} 
   \int_{|t|= 1/2}
   \frac{ t^{j-1} r_{n-j}^2 }{ ( r_n + 2t r_{n-1} ) (r_n - r_{n-1}) } dt.
\end{equation}


The next task is to study the poles of the integrand.

\begin{lemma}
The zeroes of $r_{n+1}-r_n$ are inside the circle $|t| = 1/2$, and
the zeroes of $r_{n+1}+2t r_n$ are outside it.
\end{lemma}

\begin{proof}
The rational function $f_1^{(n)}$ is analytic for $|z| \le 1/\sqrt 2$
(this is implicit in the proof of Lemma 17 of [1]), and this
holds for $f_j^{(n)}$ as well, since it has the same denominator 
as $f_1^{(n)}$.  The result is a consequence of this and the computations
used to derive (\ref{eqn:integral}).
\end{proof}

As a consequence of this lemma, we can choose an $\epsilon > 0$ so
that $p_j^{(n)}$ is given as in (\ref{eqn:integral}), but with
the contour of integration now $|t| = 1/2 -\epsilon$.

\begin{lemma}\label{div}
Let $n \ge 2$.  If $1 \le j \le n$,
$H_j=t^{j-1}(1+2t)r_{n-j} + (-1)^j(r_j-r_{j-1})(r_n+2tr_{n-1})$
is a polynomial in $t$, divisible by $r_n-r_{n-1}$.
\end{lemma}

\begin{proof}
The recurrence relation on $r_j$ implies that
$H_{j+2}=(1-2t)H_{j+1}+H_j$, so it suffices to check
the statement for $j=1$ and $j=2$. It is easy to check that
$H_1=r_{n-1}-r_n$, and the recurrence
relation on $r_n$ implies that $H_2=r_n-r_{n-1}$.
\end{proof}

The above proof has the consequence that $H_j/(r_n-r_{n-1})$ 
depends on $j$ only.

\begin{theorem}
\label{thm:main}
Let $n \ge 2$ and $1 \le j \le n$.  Then (\ref{eqn:explicit}) holds.
\end{theorem}

\begin{proof}
Using Lemma \ref{div} to rewrite $t^{j-1} r_{n-j}$,
we have for $1 \leq j \leq n$
$$
p_j^{(n)} = -\frac 1{2\pi i}
\int_{\vert t\vert=\frac 12-\epsilon}
\frac {r_{n-j}(r_j-r_{j-1}) dt}{(1+2t)(r_n-r_{n-1}) }
+ \frac 1 {2 \pi i} \int_{\vert t\vert=\frac 12-\epsilon}
\frac {R_j(t) dt}{(r_n+2tr_{n-1})(1+2t)}, 
$$
for some polynomial $R_j$.
The second integral is zero, because all the poles
of the rational function are outside the integration contour.
In the first integral, the integrand
$$
\frac {r_{n-j}(r_j-r_{j-1})}{(1+2t)(r_n-r_{n-1})},
$$
as a function on the Riemann sphere,
has a unique singularity outside the contour, which is a 
pole of order $1$ at $t = -\frac 12$.
Indeed, the degree of $r_k$ is $k-1$ so the degree of the denominator
is two plus the degree of the numerator, which assures that there is
no pole at infinity.

Since the residues of a rational function sum to zero, we get
\begin{equation}
\label{eqn:pjn}
p_{j}^{(n)}={\rm Res}_{t=-\frac 12} 
\frac {r_{n-j}(r_j-r_{j-1})}{(1+2t)(r_n-r_{n-1})} 
= \frac 1 2
\frac {r_{n-j}(r_j-r_{j-1})}{(r_n-r_{n-1})} \left( - \frac 12 \right) .
\end{equation}

To derive the explicit formula, observe that
$\alpha(-1/2) = \frac{2 + \sqrt 2} 2$ and
$\beta(-1/2) = \frac{2 - \sqrt 2} 2$, so
(with $A,B = 2 \pm \sqrt 2$)
$$
r_k (- \half) = \frac{A^k - B^k} {2 ^ {k + 1/2}},
\qquad
(r_k - r_{k-1}) (-\half) = \frac{A^{k-1} - B^{k-1}} {2 ^ {k}}.
$$
Then, substitute these values into (\ref{eqn:pjn}) and simplify.
\end{proof}

It is natural to extend our notation so that $p_0^{(n)} = 1$
for $n \ge 2$.  The explicit formula does not work there, but the
above proof indicates a reason for this.  We could use (\ref{rec})
to extend $r_j$ to $j = -1$, but then $r_{-1}$ would not be a polynomial,
invalidating our arguments.

Using the recurrence relation for $r_k$, we can prove that
when $n \ge 1$, the polynomial $r_n - r_{n-1}$ has distinct roots.
However, the computations for this are not very enlightening, so 
we leave verification of this to the reader.

\section{Remarks.}
\label{sect:computing}

Initially, we had arrived at the formula (\ref{eqn:explicit}) after numerical
calculations that were done by a different method, which we believe
to be of independent interest.  In this section, we discuss how
these calculations were done. 
The idea is to combine numerical approximation of the residues with
a bound on the denomininator.

Suppose we want a value for
\begin{equation}
\label{int:generic}
p := 
\frac 1 {2 \pi i} \int_{|z| = a} |f(z)|^2 \frac {dz} z,
\end{equation}
in which $f \in {\bf Q}(z)$, and $a \in \bf Q$.  Let $a$ be large
enough that all the poles of $f$ are inside the circle
$|z| = a$.  On this circle, $\bar z = a^2 / z$, and if
we use this to get an expression for $\bar f$, we can bring
(\ref{int:generic}) into the form
$$
p = 
\frac{m_1}{m_2} \cdot 
{1 \over 2 \pi i} \int_{|z| = a} \frac {b(z) dz}{c(z)d(z)},
$$
in which $m_1,m_2 \in \bf Z$, $b,c,d \in {\bf Z}[v]$, and $c,d$ are
monic.  The zeroes of $c$ and $d$ are algebraic integers, and we choose
notation so that those of $c$ and $d$ are outside and inside the circle,
respectively.

In the cases of interest to us, $c$ and $d$ had distinct roots,
so let us make this simplifying assumption.  Then, evaluating 
$(2 \pi i)^{-1} \int b (cd)^{-1} dz$ by residues produces
\begin{equation}
\label{eqn:residues}
\sum_\xi \frac{b(\xi)} {\prod_\eta(\xi-\eta)}
            \cdot \frac 1 { \prod_{\xi' \ne \xi} (\xi - \xi') }
\end{equation}
In this expression, $\xi$ and $\xi'$ range over zeroes of $d$,
and $\eta$ ranges over zeroes of $c$.  The denominators are algebraic
integers, and
$$
\prod_\eta(\xi-\eta) \ \mid\ \hbox{resultant}(c,d) := R;
$$
$$
\prod_{\xi'\ne\xi}(\xi-\xi') \ \mid\ \hbox{discriminant}(d) := D.
$$
Since $b$ has integral coefficients, we conclude that 
(\ref{int:generic}) is an integral multiple of $\Delta^{-1}$, where
$$
\Delta = m_2 R D .
$$
In particular (by Galois theory), $p$ is a rational number.

To obtain (\ref{int:generic}) exactly, then, it will suffice to compute
the integer $\delta$, and then evaluate (\ref{eqn:residues}) using
numerical approximations to the zeroes, with enough accuracy to 
determine (\ref{int:generic}) to the nearest integer.


Using this method, we were able to compute absorption probabilities
exactly up to $n=20$ in a couple of minutes on a workstation.  Straight
numerical integration would have been much slower, 
and would not have given us exact results.

We end this paper with a short table of the $p_j^{(n)}$.
The numerators of $p_1$ are Sequence A084068 in \cite{Sloane}.

\begin{center}
TABLE 1. Absorption Probabilities $p_j^{(n)}$.
\end{center}
$$
\begin{array}{rcccccccc}
   &  j=1   &    2  &     3   &     4  &  5     & 6      & 7      &   8    \\
n=2&  1/2   &       &         &        &        &        &        &        \\
3  &  2/3   &  1/3  &         &        &        &        &        &        \\
4  & 7/10   &  4/10 &    3/10 &        &        &        &        &        \\
5  & 12/17  &  7/17 &    6/17 &    5/17&        &        &        &        \\
6  & 41/58  &  24/58&    21/58&   20/58& 17/58  &        &        &        \\
7  & 70/99  &  41/99&    36/99&   35/99& 34/99  & 29/99  &        &        \\
8  & 239/338&140/338& 123/338 & 120/338& 119/338& 116/338& 99/338 &        \\
9  & 408/577&239/577& 210/577 & 205/577& 204/577& 203/577& 198/577& 169/577
\end{array}
$$

\section*{Acknowledgments}

The research of Eric Bach was supported by the National Science
Foundation (grants CCF-0635355 and CCF-0523680), and the Wisconsin
Alumni Research Foundation (through a Vilas Associate Award).
The hospitality of the
University of Waterloo is also gratefully acknowledged.

The research of Lev Borisov was supported by the National
Science Foundation (grant DMS-0758480).

\bibliographystyle{elsart-num}


\begin{thebibliography}{10}
\expandafter\ifx\csname url\endcsname\relax
  \def\url#1{\texttt{#1}}\fi
\expandafter\ifx\csname urlprefix\endcsname\relax\def\urlprefix{URL }\fi

\bibitem{ABNVW}
A.~Ambainis, E.~Bach, A.~Nayak, A.~Vishwanath, J.~Watrous, One-dimensional
  quantum walks, in: Proceedings of the Thirty-Third Annual ACM Symposium on
  Theory of Computing, 2001, pp. 60--69.

\bibitem{BCGJW}
E.~Bach, S.~Coppersmith, M.~Paz Goldschen, R.~Joynt, and J.~Watrous,
One-dimensional quantum walks with absorbing boundaries, 
J. Comp. Sys. Sci., 69, 2004, 562-592.

\bibitem{DS} 
P. G. Doyle and J. L. Snell, Random Walks and Electrical Networks,
MAA, 1984.

\bibitem{MF} 
M. Mend\`es France, The Rudin-Shapiro sequence, Ising chain, and
paperfolding, in B. C. Berndt et al., eds., Analytic Number
Theory: Proceedings of a Conference in Honor of Paul Bateman,
Birkh\"auser, 1990, pp. 367-382.

\bibitem{Sloane} 
N. J. A. Sloane, On-Line Encyclopedia of Integer Sequences.

\end{thebibliography}

\end{document}